\documentclass[%
 aip,
 amsmath,amssymb,
 reprint,%
]{revtex4-1}

\usepackage{graphicx}
\usepackage{dcolumn}
\usepackage{bm}

\usepackage[utf8]{inputenc}
\usepackage[T1]{fontenc}
\usepackage{mathptmx}
\usepackage{etoolbox}
\usepackage{subcaption}
\newtheorem{thm}{Theorem}
\newtheorem{rem}{Remark}
\newtheorem{lem}{Lemma}
\newtheorem{cor}{Corollary}
\newenvironment{proof}{ {\it{Proof.}} }{}

\makeatletter
\def\@email#1#2{%
 \endgroup
 \patchcmd{\titleblock@produce}
  {\frontmatter@RRAPformat}
  {\frontmatter@RRAPformat{\produce@RRAP{*#1\href{mailto:#2}{#2}}}\frontmatter@RRAPformat}
  {}{}
}%
\makeatother
\begin{document}

\preprint{AIP/123-QED}

\title[]{Tracking the vortex motion by using Brownian fluid particles}
\author{Zhongmin Qian}
\email{qianz@maths.ox.ac.uk}
\affiliation{Mathematical Institute, University of Oxford, OX2 6GG, England}
\altaffiliation[Also at ]{Oxford Suzhou Centre for Advanced Research}
 
\author{Youchun Qiu}
\affiliation{ 
Institut de Math\'ematiques de Toulouse, UMR 5219, Universit\'e
de Toulouse, CNRS, UPS, F-31062, Toulouse Cedex 9, France.
}%

\author{Yihuang Zhang}
\affiliation{Mathematical Institute, University of Oxford, OX2 6GG, England}
 \altaffiliation[Also at ]{Oxford Suzhou Centre for Advanced Research}
 
\date{\today}

\begin{abstract}
In this paper we propose a simple yet powerful vortex method to numerically
approximate the dynamics of an incompressible flow. The idea is to
sample the distribution of the initial vortices of the fluid flow
in question then follow vortex dynamics along Taylor's Brownian fluid
particles. The weak convergences of this approximation scheme are obtained
for 2D and 3D cases, though only for small time in 3D case. Based on our
method, the simulation results are quite attracting. 
\end{abstract}

\maketitle

\section{\label{sec:level1}Introduction}
Numerical methods have become important components in the study of
fluid dynamics, in particular for modeling turbulence flows. With
the advance of computational power DNS (Direct Numerical Simulation),
LES (Large Eddy Simulation) and other new technologies have been developed
in recent years for solving the Navier-Stokes equations numerically.
Among them, various vortex methods which are based on the vorticity transport
equation have become an attractive approach for simulating fluid flows
in particular turbulent flows. One may find a comprehensive account
in the monographs Cottet and Koumoutsakos\cite{cottet2000vortex}, Majda and Bertozzi\cite{majda2002vorticity}, Saffman\cite{saffman1992vortex}, Ting and Knio\cite{ting2007vortex},
and the recent review Mimeau and Mortazavi\cite{Mimeau-Mortazavi2021}.

The idea of vortex methods was introduced in Chorin\cite{chorin1973numerical}.
Several vortex approximation procedures and convergence results have
been established for 2D flows, see for example Anderson et  al.\cite{Anderson-Greengard1985,beale1982vortex,Hald1978,Hald1979,Hald1987}
and the literature therein. For 3D vortex methods of inviscid fluid flows,
the convergence with Lagrangian stretching was proved in Beale and Majda\cite{beale1982vortex}.
A different approach which updates the vorticity through the velocity
field was proposed and the corresponding convergence result was shown
in Beale\cite{Beale1986}. For viscous fluid flows, the vortex dynamics is replaced
by a random dynamical system in which a Brownian motion term is added
to the equation of motion of the fluid particles, and the fluid particles
become Brownian fluid particles. It\^{o}'s stochastic differential
equations (SDEs) take place of ordinary differential equations (ODEs).
A remarkable convergence result for 2D random vortex method has been
established in Long \cite{long1988convergence}. For 3D viscous fluid flows
however, to the best of authors' knowledge, there is no satisfactory
solution yet so far.

In the present paper, we propose a simple method of tracking the vortex
dynamics of an incompressible fluid which gives surprisingly satisfactory
simulations for both inviscid and viscous fluid flows. Our method
is to introduce the vorticity evolution directly in terms of Brownian
fluid particles specified by a set of stochastic differential equations.
We do not evolve the Taylor diffusion by mollifying the Biot-Savart
kernel as in the traditional vortex methods, but instead we sample
the distribution of the initial vortices and develop the initial distribution
according to the SDEs determined by the vorticity equation. Let us
describe this approach in more detail and at the same time establish
the notations we will use throughout the paper.

Let $u=(u^{1},u^{2},u^{3})$ denote the velocity of an incompressible
fluid flow moving in a range without boundary constraint. Hence the
velocity $u(x,t)$ satisfies the equations of motion, the Navier-Stokes
equations 
\begin{equation}
\frac{\partial}{\partial t}u^{i}+u^{j}\frac{\partial}{\partial x^{j}}u^{i}=\nu\Delta u^{i}-\frac{\partial}{\partial x^{i}}p,\quad\frac{\partial}{\partial x^{k}}u^{k}=0\label{eq:nseq1}
\end{equation}
in $\mathbb{R}^{3}$, where $i=1,2,3$, $\nu\geq0$ is the kinetic
viscosity and $p(x,t)$ is the pressure which is uniquely determined
by $u(x,t)$ up to a constant at every $t$. Einstein's convention
that the term with a pair of repeated indices are summed over from
1 to 3 has been applied. The case where $\nu=0$ corresponds to inviscid
fluid flows, and the Navier-Stokes equations are reduced to the Euler
equations.

The vorticity of $u$ is denoted by $\omega=(\omega^{1},\omega^{2},\omega^{3})$,
although, which will be specified, we will depart from this convention
in order to present some results in a general setting. By definition
$\omega^{i}=\varepsilon^{ijk}\frac{\partial}{\partial x^{j}}u^{k}$
for $i=1,2,3$. The vorticity equations, which are the equations of
vortex motion, play a dominated role in vortex methods, and are
obtained by differentiating the Navier-Stokes equations 
\begin{equation}
\frac{\partial}{\partial t}\omega^{i}+u^{j}\frac{\partial}{\partial x^{j}}\omega^{i}=\nu\Delta\omega^{i}+\omega^{j}\frac{\partial}{\partial x^{j}}u^{i}\label{vort-eq1}
\end{equation}
for $i=1,2,3$. The non-linear term $u\cdot\nabla\omega$ on the left-hand
side is the convection of $\omega$ which appears for both 2D and
3D flows. The second non-linear term $\omega\cdot\nabla u$, representing
the stretching of the vorticity, appears only in 3D flows. This fact makes
substantial difference between 2D flows and 3D flows. The study of some turbulence 
problems (see for example Saffman\cite{saffman1992vortex}, Pullin and Saffman\cite{Pullin1998}) can be formulated in terms of Cauchy's initial value problem to
the vorticity equations (\ref{vort-eq1}) together with the equation that $\omega=\nabla\wedge u$, subject to the initial vorticity $\omega_0=\omega(\cdot,0)$. 
Vortex methods aim to provide numerical schemes to the initial value problem.

The velocity $u(x,t)$, under the assumption that both $u$ and $\omega$
decay sufficiently fast at the infinity, may be recovered from reading
the vorticity via the Biot-Savart law 
\begin{equation}
u^{i}(x,t)=\int_{\mathbb{R}^{3}}\varepsilon^{ijk}G^{j}(x-y)\omega^{k}(y,t)dy\label{eq:vort-eq2}
\end{equation}
where $G=(G^{1},G^{2},G^{3})$ and 
\[
G(z)=-\frac{1}{4\pi}\frac{z}{|z|^{3}}
\]
is the Biot-Savart singular integral kernel.

Now we are in a position to describe our simple random vortex dynamics.
Two (random) vector fields $V(x,t)$ and $W(x,t)$ will be defined below which do
not necessarily satisfy the Navier-Stokes equations nor the relation
that $W=\nabla\wedge V$. Our goal is in fact to construct approximation
solutions to the vorticity equation (\ref{vort-eq1}) directly, in
the spirit which is quite like Feynman's functional integration for
Schr\"{o}dinger's equations. Hence $V$ will be the approximate velocity
of $u$, and $W$ the approximation of $\omega$.

Following the general ideas in the vortex methods, we propose the
following dynamics scheme for the vortex motion of an incompressible fluid flow. At the initial time
$t=0$, we sample a (finite) collection of locations $x_{n}\in\mathbb{R}^{3}$,
where $n$ runs through a finite index set, at which the major vortex
motion may be demonstrated. For example $x_{n}$ can be the centers of
vortex rings. Suppose the initial vorticity of the fluid flow is given
by 
\begin{equation}
W^{i}(x,0)=\sum_{n}A_{n}^{i}(0)\varphi(x-x_{n})\label{eq:int-dis0}
\end{equation}
where $A_{n}(0)$, vectors with components $A_{n}^{i}(0)$ (for $i=1,2,3$)
are the initial vortices, and $\varphi$ is a wavelet type function which should be
close to the Dirac delta (at $0$) function. Hence at the initial stage, the vortices of the fluid flow are distributed among $x_n$ so that
\[ W(x,0) \simeq \sum_{n}A_{n}(0)\delta(x-x_{n}),\]
see for example Cottet and Koumoutsakos\cite{cottet2000vortex}. Therefore (\ref{eq:int-dis0})
provides us with the distribution of the initial vortices in the fluid flow. The key
idea is that both the loci $x_{n}$ and the initial vortices $A_{n}(0)$
are sampled so that they are representative for the vortex motion
at the initial stage. The initial vorticity $W(x,0)$ will
be transported along the Brownian fluid particles. More precisely
the vorticity $A_{n}(0)$ at time $t>0$ is transported to a new location
$X_{n}(t)$ with a new vorticity $A_{n}(t)=(A_{n}^{1}(t),A_{n}^{2}(t),A_{n}^{3}(t))$,
so that the distribution of the vortices at time $t>0$ is given by
\begin{equation}
W^{i}(x,t)=\sum_{n}A_{n}^{i}(t)\varphi(x-X_{n}(t))\label{eq:vort-2.2}
\end{equation}
for $i=1,2,3$. The velocity $V(x,t)$ at time $t>0$ is defined in
terms of the Biot-Savart law (\ref{eq:vort-eq2}): 
\begin{equation}
V^{i}(x,t)=\int_{\mathbb{R}^{3}}\varepsilon^{ijk}G^{j}(x-y)W^{k}(y,t)dy,\label{eq:vort-eq2-1}
\end{equation}
so that the relation between $W$ and $V$ is maintained at least
partly. It remains to determine the dynamics of $A_{n}$ and $X_{n}$.
By initiative, the velocity of $X_{n}=(X_{n}^{1},X_{n}^{2},X_{n}^{3})$
should be the velocity of the fluid flow, so that 
\begin{equation}
dX_{n}^{i}(t)=V^{i}(X_{n}(t),t)dt+\sqrt{2\nu}dB^{i}(t),\quad X_{n}(0)=x_{n}\label{eq:taylor-01}
\end{equation}
for $i=1,2,3$ and $B=(B^1,B^2,B^3)$ is a standard 3D Brownian motion on a probability space $(\Omega,\cal{F},\mathbb{P})$. That is, $X_{n}$ are Taylor's diffusions initiated
from $x_{n}$, and therefore $X_{n}$ are the Brownian fluid particles
started at $x_{n}$. The dynamics of the vortices $A_{n}=(A_{n}^{1},A_{n}^{2},A_{n}^{3})$
are determined by the following ordinary differential equations 
\begin{equation}
dA_{n}^{j}(t)=A_{n}^{k}(t)\frac{\partial V^{j}}{\partial x^{k}}(X_{n}(t),t)dt,\quad A_{n}(0)=A_{n}(0)\label{eq:stratching1}
\end{equation}
where $i=1,2,3$, which are responsible for the vortex stretching
of the fluid flow. For 2D fluid flows, since there is no vorticity
stretching, so $A_{n}$ stay as constant vectors along the fluid particle
trajectories, see the section about 2D flows below.

In the next section, we will translate the previous system (\ref{eq:taylor-01}),
(\ref{eq:stratching1}), (\ref{eq:vort-2.2}) and (\ref{eq:vort-eq2-1})
together into a closed system of stochastic differential equations
which thus define our random vortex method.

The main contribution of the present paper is to show that the previous
simple vortex dynamics gives rise to good approximations to the motions
of vortex dynamics. We demonstrate this by showing several theoretical
results about the approximation solutions to the vorticity equations
constructed in terms of $X_{n}$ and $A_{n}$, and also by simulations
based on this simple vortex method.

The paper is organized as the following. In the next Section 2, we describe
the system of stochastic differential equations which implement a
simple vortex method. This system of SDEs
allows us to employ the Monte-Carlo simulation to the study of incompressible
fluid flows, which will be demonstrated in Section 7. In Section 3,
we derive the approximation vorticity equation, which takes a form
of a simple stochastic partial differential equation. In Section 4
and 5 we discuss the weak convergence results with respect to certain
Sobolev norms, which justify our simple vortex method. In Section
6, we discuss the 2D incompressible flows, and not surprisingly we
show that weak convergence holds for all time.

\section{Description of the random vortex method}

In this section we propose, following the approach outlined in the
Introduction, a random vortex dynamics system in a slightly general setting,
otherwise we will maintain the notations established in the previous
section.

The main change will be made for the relation between the vorticity
field $W$ and the velocity $V$ which may be not given by the vector
identity $W=\nabla\wedge V$, instead it will be given in terms of
a vector integral kernel $K=(K^{1},K^{2},K^{3})$, where $K^{i}$
are locally integrable functions on $\mathbb{R}^{3}$. The reason
to work with a slightly general kernel $K$ than the Biot-Savart kernel
$G$ is the following. For solving the initial value problem to the
vorticity equations (\ref{vort-eq1}) where $\omega=\nabla\wedge u$
and $\textrm{div}u=0$, which are equivalent to the system (\ref{vort-eq1})
and (\ref{eq:vort-eq2}), mathematical difficulty arises since the
Biot-Savart kernel $G$ is singular, so it is natural to replace $G$
by its smooth approximations. The simple way to create an approximation is
to mollify the Biot-Savart kernel $G$. That is, choosing a smooth
function $\psi\geq0$ with a compact support in $(-\frac{1}{2},\frac{1}{2})^{3}$
and $\int_{\mathbb{R}^{3}}\psi(x)dx=1$. For each $\delta>0$, set
$\psi_{\delta}(x)=\delta^{-3}\psi(\delta^{-1}x)$ and $G_{\delta}(x)=G\star\psi_{\delta}(x)$,
where $\star$ denotes the convolution, i.e. 
\begin{equation}
G_{\delta}(x)=\int_{\mathbb{R}^{3}}G(y)\psi_{\delta}(x-y)dy.\label{G-delta-def}
\end{equation}
Then $G_{\delta}$ is smooth and 
\begin{equation}
\left\Vert D^{k}G_{\delta}\right\Vert \leq C_{k}\frac{1}{\delta^{2+k}}\label{Est-G-delta}
\end{equation}
for some constant $C_{k}>0$, for $k=1,2,\cdots$, depending only on the regularization $\psi$. Moreover
$G_{\delta}\rightarrow G$ as $\delta\downarrow0$ in distribution
sense. Therefore the initial value problem to the following system
\begin{equation}
\frac{\partial}{\partial t}\omega^{i}+u^{j}\frac{\partial}{\partial x^{j}}\omega^{i}=\nu\Delta\omega^{i}+\omega^{j}\frac{\partial}{\partial x^{j}}u^{i}\label{a-vort1}
\end{equation}
with the initial data $\omega(x,0)=\omega_0(x)$, and 
\begin{equation}
u^{i}(x,t)=\int_{\mathbb{R}^{3}}\varepsilon^{ijk}G_{\delta}^{j}(x-y)\omega^{k}(y,t)dy\label{a-vort2}
\end{equation}
for $\delta>0$ gives rise to approximation solutions, denoted by $\omega^{\delta}$
and $u^{\delta}$. In fact as long as the initial data $\omega_{0}$
is smooth with a compact support, then $\omega^{\delta}\rightarrow\omega$
as $\delta\downarrow0$ with respect to some Sobolev norm at least for small time.  Therefore, from the computational view-point, we only need to develop numerical schemes for the approximation solutions $\omega^{\delta}$
and $u^{\delta}$, thus
it is important to work with singular kernels $K$ such as $G_{\delta}$.
While we would like to point out that  the procedure for going to the approximation equations (\ref{a-vort1})
and (\ref{a-vort2}) seems unnecessary for implementing the random
vortex method below, rather than for the technical reason that a priori estimates are
not available for 3D Navier-Stokes equations, see Lemma \ref{lem4}
below.

Let us now define the random vortex system. In defining the distribution
of the initial vortices (\ref{eq:int-dis0}) the wavelet type function
$\varphi$ is assumed to be smooth with a compact support about the
original $0$. Let us assume that the support of $\varphi\geq0$ lies
inside the ball at $0$ with radius $r_{\varphi}>0$ which will be
chosen to be small, and the total mass $\int_{\mathbb{R}^{3}}\varphi(x)dx=1$.
These are the structure data for our random vortex scheme, which are
fixed if we are not care about the convergence issue.

Let $W(x,0)$ and $W(x,t)$ (for $t>0$) be given by (\ref{eq:int-dis0})
and (\ref{eq:vort-2.2}) respectively, where $x_{n}$ and $A_{n}(0)$
are sampled initially, so they are the fixed data too. We then define
the vector field $V(x,t)$ by the convolution (together with the wedge
product) 
\begin{equation}
V^{i}(x,t)=\int_{\mathbb{R}^{3}}\varepsilon^{ijk}K^{j}(x-y)W^{k}(y,t)dy\label{eq:v-k-def1}
\end{equation}
which coincides with (\ref{eq:vort-eq2}) if $K$ is the Biot-Savart
kernel $G$. The dynamics of $(X_{n},A_{n})$ are still defined by
(\ref{eq:taylor-01}) and (\ref{eq:stratching1}). Together with (\ref{eq:vort-2.2})
we deduce that 
\begin{equation}
V^{i}(x,t)=-\sum_{n}\varepsilon^{ikj}A_{n}^{k}(t)K_{\varphi}^{j}(x-X_{n}(t))\label{eq:vector-01}
\end{equation}
for $i=1,2,3$, $t\geq0$ and $x\in\mathbb{R}^{3}$, where 
\[
K_{\varphi}^{j}(x)=\int_{\mathbb{R}^{3}}K^{j}(x-y)\varphi(y)dy
\]
(where $j=1,2,3$) turns out to be the mollification of $K^{j}$ by
 $\varphi$. This is a nice feature in this scheme.
By utilizing equation (\ref{eq:vector-01}) and substituting it into
(\ref{eq:taylor-01}) the dynamic system for the Brownian fluid particles
$X_{n}$ can be reformulated as the following SDEs 
\begin{eqnarray}
&&dX_{n}^{i}(t)=-\sum_{m}\varepsilon^{ikj}A_{m}^{k}(t)K_{\varphi}^{j}(X_{n}(t)-X_{m}(t))dt+\sqrt{2\nu}dB^{i}(t), \nonumber
\\&& X_{n}(0)=x_{n},\label{eq:X-n eq2}
\end{eqnarray}

and similarly the dynamics system (\ref{eq:stratching1}) for $A_{n}$
may be written as the following ODEs 
\begin{eqnarray}
&&dA_{n}^{i}(t)=A_{n}^{l}(t)\sum_{m}\varepsilon^{ijk}\frac{\partial K_{\varphi}^{j}}{\partial x^{l}}(X_{n}(t)-X_{m}(t))A_{m}^{k}(t)dt, \nonumber\\&& A_{n}(0)=A_{n}(0),\label{eq:A-n eq2}
\end{eqnarray}
where $i=1,2,3$, $m$ and $n$ run through the finite range of the
initial locations $x_{n}$.

The system of SDEs (\ref{eq:X-n eq2}, \ref{eq:A-n eq2}) is closed
and depends only on the structure data $K$ and $\varphi$. We observe
that $K_{\varphi}$ is smooth if $K$ is locally integrable under
our assumption on $\varphi$. 
\begin{thm}
\label{thm1}Given a finite collection of loci $x_{n}\in\mathbb{R}^{3}$
and a family of initial vortices $A_{n}(0)$, there is a unique maximal
strong solution $(X_{n},A_{n})$ to SDEs (\ref{eq:X-n eq2}, \ref{eq:A-n eq2})
up to the explosion time $\tau>0$. The vector fields defined by 
\begin{equation}
W^{i}(x,t)=\sum_{n}A_{n}^{i}(t)\varphi(x-X_{n}(t))\label{sol-W-1}
\end{equation}
and 
\begin{equation}
V^{i}(x,t)=-\sum_{n}\varepsilon^{ikj}A_{n}^{k}(t)K_{\varphi}^{j}(x-X_{n}(t))\label{sol-V-1}
\end{equation}
are smooth in $x$ for $t<\tau$, where $i=1,2,3$. Moreover the explosion
time $\tau\geq\tau_{A(0),\varphi}$, where 
\begin{equation}
\tau_{A(0),\varphi}=\frac{1}{\left\Vert A(0)\right\Vert \left\Vert DK_{\varphi}\right\Vert _{\infty}},\label{expl-01}
\end{equation}
$\left\Vert A\right\Vert =\sum_{n}|A_{n}|$ and $\left\Vert f\right\Vert _{\infty}$
denotes the $L^{\infty}$-norm of a function on $\mathbb{R}^{3}$.
We note that $\tau_{A(0),\varphi}>0$ is deterministic. 
\end{thm}

\begin{proof}
The proof follows from the standard result in It\^{o}'s theory of stochastic
differential equations, see for example Theorem 2.3 on page 173 in
Ikeda and Watanabe\cite{IkedaWatanabe}. The coefficients in defining SDEs (\ref{eq:X-n eq2})
and (\ref{eq:A-n eq2}) are in general not globally Lipschitz, but
nevertheless locally Lipschitz continuous. Therefore the explosion
time $\tau$ may be finite, and the maximal strong solution $(X_{n},A_{n})$
is unique for a given 3D Brownian motion $B=(B^{1},B^{2},B^{3})$
on a probability space $(\varOmega,\mathcal{F},\mathbb{P})$. To derive
the estimate (\ref{expl-01}), we utilize a truncation technique.
For each $L>0$, let $\theta_{L}(s)$ be the cutting-off function
which equals $s$ if $|s|\leq L$ and equals $L$ if $|s|>L$. The
unique strong solution $(X^{L},A^{L})$ to the truncated SDEs 
\[
dX_{n}^{i}(t)=-\sum_{m}\varepsilon^{ikj}\theta_{L}(A_{m}^{k}(t))K_{\varphi}^{j}(X_{n}(t)-X_{m}(t))dt+\sqrt{2\nu}dB^{i}(t)
\]
and 
\[
dA_{n}^{i}(t)=\theta_{L}(A_{n}^{l}(t))\sum_{m}\varepsilon^{ijk}\frac{\partial K_{\varphi}^{j}}{\partial x^{l}}(X_{n}(t)-X_{m}(t))\theta_{L}(A_{m}^{k}(t))dt
\]
subject to the same initial data exist for all $t$ (Theorem 2.4,
page 177 in Ikeda and Watanabe\cite{IkedaWatanabe}). It follows from the second equation
that 
\begin{equation}
\left\Vert A_{n}^{L}(t)\right\Vert \leq\left\Vert A_{n}^{L}(0)\right\Vert +\int_{0}^{t}\left\Vert A_{n}^{L}(s)\right\Vert \sum_{m}\left\Vert DK_{\varphi}\right\Vert _{\infty}\left\Vert A_{m}^{L}(s)\right\Vert ds,\label{Ant-ineq1}
\end{equation}
by adding the inequalities where $n$ runs through the index set,
to obtain that 
\begin{equation}
\left\Vert A^{L}(t)\right\Vert \leq\left\Vert A(0)\right\Vert +\left\Vert DK_{\varphi}\right\Vert _{\infty}\int_{0}^{t}\left\Vert A^{L}(s)\right\Vert ^{2}ds\label{At-ineq1}
\end{equation}
for $t>0$. By Gronwall's inequality 
\begin{equation}
\left\Vert A^{L}(t)\right\Vert \leq\frac{\left\Vert A(0)\right\Vert }{1-\left\Vert A(0)\right\Vert \left\Vert DK_{\varphi}\right\Vert _{\infty}t}\label{At-ineq2}
\end{equation}
for all $t<\tau_{A(0),\varphi}$. The estimate (\ref{expl-01}) now
follows immediately by sending $L\uparrow\infty$. 
\end{proof}
\begin{rem}
From  (\ref{At-ineq2}), it is clear that
\begin{equation}
    \lVert A(t)\rVert \leq 2\lVert A(0)\rVert \label{At-rem}
\end{equation}
for $t<\frac{1}{2}\tau_{A(0),\varphi}$. 

By Gronwall's inequality, it follows from (\ref{Ant-ineq1}) that 
\begin{equation}
\left\Vert A_{n}(t)\right\Vert \leq\left\Vert A_{n}(0)\right\Vert e^{\left\Vert DK_{\varphi}\right\Vert _{\infty}\int_{0}^{t}\left\Vert A(s)\right\Vert ds}\label{Aest3.6}
\end{equation}
for $t<\tau_{A(0),\varphi}$, and therefore, by (\ref{At-rem})
\begin{equation}
\left\Vert A_{n}(t)\right\Vert \leq\left\Vert A_{n}(0)\right\Vert e^{2\left\Vert DK_{\varphi}\right\Vert _{\infty}\left\Vert A(0)\right\Vert t}\label{Aest3.7}
\end{equation}
for $t<\frac{1}{2}\tau_{A(0),\varphi}$. 
\end{rem}

In the next section, we demonstrate that $(W,V)$ defined in Theorem
\ref{thm1} is an approximation solution to the vorticity equation
(\ref{vort-eq1}) in certain sense.

\section{Approximating the vorticity equation}

In this section, we assume that the structure data $K,\varphi$, $x_{n}$
and $A_{n}(0)$ are given as in the previous section. $(X_{n},A_{n})$
is the unique maximal solution pair to the SDEs (\ref{eq:X-n eq2})
and (\ref{eq:A-n eq2}), and $V(x,t)$ and $W(x,t)$ are defined in
Theorem \ref{thm1}. 
\begin{thm}
\label{thm2}For each $x\in\mathbb{R}^{3}$, $W(x,t)$ and $V(x,t)$
are continuous semi-martingales for $t\in[0,\tau_{A(0),\varphi})$,
and 
\begin{align}
dW^{i} & =\left(W^{j}\frac{\partial V^{i}}{\partial x^{j}}-V^{j}\frac{\partial W^{i}}{\partial x^{j}}+\nu\Delta W^{i}\right)dt\nonumber \\
 & -\sqrt{2\nu}\frac{\partial W^{i}}{\partial x^{j}}dB^{j}(t)+\left(F^{i}+G^{i}\right)dt\label{approx-vort1}
\end{align}
for $t\in[0,\tau_{A(0),\varphi}]$, where the error terms 
\begin{equation}
F^{i}(x,t)=\sum_{n}A_{n}^{i}(t)\left(V^{j}(x,t)-V^{j}(X_{n}(t),t)\right)\frac{\partial\varphi}{\partial x^{j}}(x-X_{n}(t))\label{F-5.9}
\end{equation}
and 
\begin{equation}
G^{i}(x,t)=\sum_{n}A_{n}^{j}(t)\left(\frac{\partial V^{i}}{\partial x^{j}}(X_{n}(t),t)-\frac{\partial V^{i}}{\partial x^{j}}(x,t)\right)\varphi(x-X_{n}(t)).\label{G-5.10}
\end{equation}
Moreover, according to (\ref{sol-V-1}) and (\ref{sol-W-1}), $V=K\star W$
where the convolution is made with the wedge product of two vectors. 
\end{thm}

\begin{proof}
Since $A_{n}(x,t)$ are continuous semi-martingales with finite variations
(in $t$), by applying It\^{o}'s formula to $W^{i}$ define in (\ref{sol-W-1}),
we have 
\begin{align*}
&dW^{i}(x,t) \\ =&\sum_{n}\varphi(x-X_{n}(t))dA_{n}^{i}(t)-\sum_{n}A_{n}^{i}(t)\frac{\partial\varphi}{\partial x^{j}}(x-X_{n}(t))dX_{n}^{j}(t)\\
+&\nu\sum_{n}A_{n}^{i}(t)\Delta\varphi(x-X_{n}(t))dt.
\end{align*}
Next we substitute $dX_{n}$ and $dA_{n}$ by using the SDEs (\ref{eq:X-n eq2})
and (\ref{eq:A-n eq2}) to obtain that 
\begin{align}
&dW^{i}(x,t)\nonumber\\  =&\sum_{n}\sum_{m}\varepsilon^{iqk}A_{n}^{l}(t)A_{m}^{k}(t)\varphi(x-X_{n}(t))\frac{\partial K_{\varphi}^{q}}{\partial x^{l}}(X_{n}(t)-X_{m}(t))dt\nonumber \\
 +&\sum_{n}\sum_{m}\varepsilon^{jkq}A_{m}^{k}(t)A_{n}^{i}(t)K_{\varphi}^{q}(X_{n}(t)-X_{m}(t))\frac{\partial\varphi}{\partial x^{j}}(x-X_{n}(t))dt\nonumber \\
 +&\nu\sum_{n}A_{n}^{i}(t)\Delta\varphi(x-X_{n}(t))dt\nonumber\\-&\sqrt{2\nu}\sum_{n}A_{n}^{i}(t)\frac{\partial\varphi}{\partial x^{j}}(x-X_{n}(t))dB^{j}(t).\label{eq:DW}
\end{align}
On the other hand, according to the construction (\ref{sol-W-1}),
(\ref{sol-V-1}) we have 
\begin{align*}
&V^{j}\frac{\partial W^{i}}{\partial x^{j}}(x,t) \\ =&-\sum_{m}\sum_{n}\varepsilon^{jkq}A_{m}^{k}(t)A_{n}^{i}(t)K_{\varphi}^{q}(x-X_{m}(t))\frac{\partial\varphi}{\partial x^{j}}(x-X_{n}(t))\\
  =&-\sum_{m}\sum_{n}\varepsilon^{jkq}A_{m}^{k}(t)A_{n}^{i}(t)K_{\varphi}^{q}(X_{n}(t)-X_{m}(t))\frac{\partial\varphi}{\partial x^{j}}(x-X_{n}(t))\\
  +&\sum_{n}A_{n}^{i}(t)\left(V^{j}(x,t)-V^{j}(X_{n}(t),t)\right)\frac{\partial\varphi}{\partial x^{j}}(x-X_{n}(t)),
\end{align*}
\begin{align*}
&W^{j}\frac{\partial V^{i}}{\partial x^{j}}(x,t) \\  =&-\sum_{n}\sum_{m}\varepsilon^{ikq}A_{m}^{j}(t)A_{n}^{k}(t)\varphi(x-X_{m}(t))\frac{\partial K_{\varphi}^{q}}{\partial x^{j}}(x-X_{n}(t))\\
   =&\sum_{n}\sum_{m}\varepsilon^{iqk}A_{n}^{l}(t)A_{m}^{k}(t)\varphi(x-X_{n}(t))\frac{\partial K_{\varphi}^{q}}{\partial x^{l}}(x-X_{m}(t))\\
   =&\sum_{n}\sum_{m}\varepsilon^{iqk}A_{n}^{l}(t)A_{m}^{k}(t)\varphi(x-X_{n}(t))\frac{\partial K_{\varphi}^{q}}{\partial x^{l}}(X_{n}(t)-X_{m}(t))\\
   -&\sum_{n}A_{n}^{l}(t)\left(\frac{\partial V^{i}}{\partial x^{l}}(X_{n}(t),t)-\frac{\partial V^{i}}{\partial x^{l}}(x,t)\right)\varphi(x-X_{n}(t))
\end{align*}
and 
\[
\Delta W^{i}(x,t)=\sum_{n}A_{n}^{i}(t)\Delta\varphi(x-X_{n}(t)).
\]
Substituting these equations into (\ref{eq:DW}), we obtain (\ref{approx-vort1}). 
\end{proof}

We next show that $W$ and $V$ are approximation solutions to the
vorticity equation (\ref{vort-eq1}) with the initial vorticity $W(x,0)$. 
\begin{lem}
\label{lem4}Suppose the support of $\varphi$ is contained in the
ball centered at $0$ with radius $r_{\varphi}>0$.

1) It holds that 
\[
|F(x,t)|\leq2\left\Vert A(t)\right\Vert ^{2}\left\Vert DK_{\varphi}\right\Vert _{\infty}\left\Vert D\varphi\right\Vert _{\infty}r_{\varphi}
\]
and 
\[
|G(x,t)|\leq2\left\Vert A(t)\right\Vert ^{2}\left\Vert D^{2}K_{\varphi}\right\Vert _{\infty}\left\Vert \varphi\right\Vert _{\infty}r_{\varphi}
\]
for $t<\tau_{A(0),\varphi}$ and $x\in\mathbb{R}^{3}$.

2) There is a positive constant $C$ depending only on $\left\Vert A(0)\right\Vert $,
$\left\Vert DK_{\varphi}\right\Vert _{\infty}$, $\left\Vert D^{2}K_{\varphi}\right\Vert _{\infty}$,
$\left\Vert D\varphi\right\Vert _{\infty}$ and $\left\Vert \varphi\right\Vert _{\infty}$
such that 
\[
|F(x,t)|\leq Cr_{\varphi}\quad\textrm{ and }|G(x,t)|\leq Cr_{\varphi}
\]
for all $t\in\left[0,\frac{1}{2}\tau_{A(0),\varphi}\right]$ and $x\in\mathbb{R}^{3}$. 
\end{lem}

\begin{proof}
By definition (\ref{sol-V-1}) 
\begin{equation}
V^{i}(x,t)=-\sum_{n}\varepsilon^{ikj}A_{n}^{k}(t)K_{\varphi}^{j}(x-X_{n}(t))\label{V-vect},
\end{equation}
we have
\begin{equation}
|V(x,t)-V(y,t)|\leq2\left\Vert A(t)\right\Vert \left\Vert DK_{\varphi}\right\Vert |x-y|\label{eq:Vd-1}
\end{equation}
and 
\begin{equation}
|DV(x,t)-DV(y,t)|\leq2\left\Vert A(t)\right\Vert \left\Vert D^{2}K_{\varphi}\right\Vert |x-y|.\label{eq:V-d2}
\end{equation}
Therefore 1) follows immediately and 2) follows from 1) and (\ref{At-rem}). 
\end{proof}
\begin{cor}
\label{cor5}Suppose the support of the wavelet function $\varphi$
is contained in the ball centered at $0$ with radius $r_{\varphi}>0$.
Let 
\begin{align}
&U^{i}(x,t)\nonumber\\=&W^{i}(x,t)-W^{i}(x,0) \nonumber\\-&\int_{0}^{t}\left(W^{j}\frac{\partial V^{i}}{\partial x^{j}} V^{j}\frac{\partial W^{i}}{\partial x^{j}}+\nu\Delta W^{i}\right)(x,s)ds\label{U-eq1}
\end{align}
for $i=1,2,3$ and $t<\tau_{A(0),\varphi}$ and $x\in\mathbb{R}^{3}$.
Then 
\begin{equation}
\left|\mathbb{E}\left[U^{i}(x,t)\right]\right|\leq C_{1}\tau_{A(0),\varphi}r_{\varphi}\label{U-eq2}
\end{equation}
for all $t\in\left[0,\frac{1}{2}\tau_{A(0),\varphi}\right]$. 
\end{cor}

\begin{proof}
By (\ref{approx-vort1}) we have 
\[
U^{i}(x,t)=-\sqrt{2\nu}\int_{0}^{t}\frac{\partial W^{i}}{\partial x^{j}}dB^{j}(s)+\int_{0}^{t}\left(F^{i}+G^{i}\right)ds.
\]
Since under the assumptions on $\varphi$, $\int_{0}^{t}\frac{\partial W^{i}}{\partial x^{j}}dB^{j}(s)$
are martingales for $t\in\left[0,\frac{1}{2}\tau_{A(0),\varphi}\right]$
so that 
\[
\mathbb{E}\left[U^{i}(x,t)\right]=\mathbb{E}\int_{0}^{t}\left(F^{i}+G^{i}\right)ds
\]
for $i=1,2,3$ and (\ref{U-eq2}) follows from Lemma \ref{lem4},
2). 
\end{proof}
\begin{rem}
If $K=G$ is the Biot-Savart kernel, then $\textrm{div}V=0$ according
to the formula (\ref{sol-V-1}) and 
\begin{equation}
V^{i}(x,t)=\int_{\mathbb{R}^{3}}\varepsilon^{ijk}G^{j}(x-y)W^{k}(y,t)dy\label{eq:v-w}
\end{equation}
which implies that $\Delta V=-\nabla\wedge W$. Therefore $W=\nabla\wedge V+\nabla f$
for some scalar function $f$, and $V$ and $W$ satisfy the vorticity
equation approximately in the sense stated in Corollary \ref{cor5}.
The equation that $W=\nabla\wedge V$ may fail in general due to the fact that $W(x,0)$ may be not divergence-free.
The reason why the approximation $V$ and $W$ still do the job (see
the simulations below) nicely is that in practice $\varphi$ is close
to an indicator function of a small ball, and therefore $W(x,t)$
is nearly divergence-free, and the relation $W\simeq\nabla\wedge V$
may be restored approximately. 
\end{rem}

\section{Weak convergence}

A simple procedure of sampling the initial vortices, for the sake
of theoretical study, may be described as the following. Suppose $\omega_{0}$
is the initial vorticity, which is smooth with a compact support $S$,
to the initial value problem to the vorticity equations (\ref{vort-eq1}).
Let $\psi$ be described at the beginning of the previous section
which defines $G_{\delta}$ for every $\delta>0$. 
\begin{lem}
\label{lem7} Let $p\geq1$, and the error terms in the random vortex
scheme, $F(x,t)$ and $G(x,t)$ be defined by (\ref{F-5.9}) and (\ref{G-5.10}).

1) The following two estimates hold: 
\begin{equation}
\left\Vert F(\cdot,t)\right\Vert _{W^{-1,p}}\leq6\left\Vert A(t)\right\Vert ^{2}\left(\left\Vert DK_{\varphi}\right\Vert \left\Vert \varphi\right\Vert _{L^{p}}r_{\varphi}+\left\Vert \varphi\nabla\wedge K_{\varphi}\right\Vert _{L^{p}}\right)\label{F-est-01}
\end{equation}
and 
\begin{equation}
\left\Vert G(x,t)\right\Vert _{L^{p}}\leq2\left\Vert A(t)\right\Vert ^{2}\left\Vert D^{2}K_{\varphi}\right\Vert _{\infty}\left\Vert \varphi\right\Vert _{L^{p}}r_{\varphi}\label{G-est-01}
\end{equation}
for all $t\in[0,\frac{1}{2}\tau_{A(0),\varphi}]$, where $\left\Vert \cdot\right\Vert _{W^{s,p}}$
denote the Sobolev norms, see \cite{Adams1975}

2) If the singular kernel $K=G_{\delta}$, where $\delta>0$, given
in (\ref{G-delta-def}), then there are universal constant $C_{1},C_{2}$
such that 
\[
\left\Vert F(\cdot,t)\right\Vert _{W^{-1,p}}\leq C_{2}\frac{r_{\varphi}}{\delta^{3}}e^{C_{1}\frac{2}{\delta^{3}}\left\Vert A(0)\right\Vert }\left\Vert A(0)\right\Vert ^{2}\left\Vert \varphi\right\Vert _{L^{p}}
\]
and 
\[
\left\Vert G(x,t)\right\Vert _{L^{p}}\leq C_{2}\frac{r_{\varphi}}{\delta^{4}}e^{C_{1}\frac{2}{\delta^{3}}\left\Vert A(0)\right\Vert }\left\Vert A(0)\right\Vert ^{2}\left\Vert \varphi\right\Vert _{L^{p}}
\]
for $t\in\left[0,\frac{1}{2}\tau_{A(0),\varphi}\wedge1\right]$. 
\end{lem}

\begin{proof}
It is clear that from definition 
\[
\left\Vert D^{k}K_{\varphi}\right\Vert _{\infty}\leq\left\Vert D^{k}K\right\Vert _{\infty}.
\]
In particular, if $K=G_{\delta}$, then 
\[
\left\Vert D^{k}K_{\varphi}\right\Vert _{\infty}\leq C_{k}\frac{1}{\delta^{2+k}}
\]
for some universal constants $C_{1},C_{2},\ldots$. Consider the linear
functional $f_{i}(h)=\int_{\mathbb{R}^{3}}F^{i}(x,t)h(x)dx$ where
$h$ defined on $\mathbb{R}^{3}$is smooth with a compact support,
i.e. a test function. Since 
\begin{equation}
F^{i}(x,t)=\sum_{n}A_{n}^{i}(t)\left(V^{j}(x,t)-V^{j}(X_{n}(t),t)\right)\frac{\partial\varphi}{\partial x^{j}}(x-X_{n}(t))\label{F-eq}
\end{equation}
so that 
\begin{align*}
&f_{i}(h) \\  =&-\sum_{n}A_{n}^{i}(t)\int_{\mathbb{R}^{3}}\left(V^{j}(x,t)-V^{j}(X_{n}(t),t)\right)\varphi(x-X_{n}(t))\frac{\partial h(x)}{\partial x^{j}}dx\\
   -&\int_{\mathbb{R}^{3}}\textrm{div}V(x,t)h(x)\varphi(x-X_{n}(t))\sum_{n}A_{n}^{i}(t)dx.
\end{align*}
By the definition for $V$ (see (\ref{sol-V-1})) one has the elementary
estimate (see the proof of Lemma \ref{lem4}) 
\[
|V^{j}(x,t)-V^{j}(y,t)|\leq2\left\Vert A(t)\right\Vert \left\Vert DK_{\varphi}\right\Vert |x-y|.
\]
Since 
\begin{align*}
\textrm{div}V(x,t) & =-\sum_{n}\varepsilon^{ikj}A_{n}^{k}(t)\frac{\partial}{\partial x^{i}}K_{\varphi}^{j}(x-X_{n}(t))\\
 & =\sum_{n}A_{n}^{k}(t)\nabla\wedge K_{\varphi}(x-X_{n}(t))
\end{align*}
so that 
\[
|\textrm{div}V(x,t)|\leq\left\Vert A(t)\right\Vert |\nabla\wedge K_{\varphi}(x-X_{n}(t))|.
\]
By using these estimates and the assumption that the support of $\varphi$
lies in the ball centered at $0$ with radius $r_{\varphi}$, we therefore
deduce that 
\begin{align*}
\left|f_{i}(h)\right| & \leq2\left\Vert A(t)\right\Vert ^{2}\left\Vert DK_{\varphi}\right\Vert r_{\varphi}\int_{\mathbb{R}^{3}}\varphi(x-X_{n}(t))|Dh(x)|dx\\
 & +\left\Vert A(t)\right\Vert ^{2}\int_{\mathbb{R}^{3}}|\nabla\wedge K_{\varphi}(x-X_{n}(t))||\varphi(x-X_{n}(t))|h(x)|dx\\
 & \leq2\left\Vert A(t)\right\Vert ^{2}\left\Vert DK_{\varphi}\right\Vert \left\Vert \varphi\right\Vert _{L^{p}}r_{\varphi}\left\Vert Dh\right\Vert _{L^{q}}\\
 & +\left\Vert A(t)\right\Vert ^{2}\left\Vert \varphi\nabla\wedge K_{\varphi}\right\Vert _{L^{p}}\left\Vert h\right\Vert _{L^{q}}
\end{align*}
where $p\geq1$ and $\frac{1}{p}+\frac{1}{q}=1$, for every test function
$h$ and $i=1,2,3$. Therefore 
\[
\left\Vert F(\cdot,t)\right\Vert _{W^{-1,p}}\leq6\left\Vert A(t)\right\Vert ^{2}\left(\left\Vert DK_{\varphi}\right\Vert \left\Vert \varphi\right\Vert _{L^{p}}r_{\varphi}+\left\Vert \varphi\nabla\wedge K_{\varphi}\right\Vert _{L^{p}}\right).
\]
The estimate for the error term $G^{i}$ is trivial. In fact by (\ref{G-5.10})
and (\ref{eq:V-d2}) to obtain 
\begin{align*}
&\left\Vert G^{i}(x,t)\right\Vert _{L^{p}} \\ \leq&\left\Vert \sum_{n}A_{n}^{j}(t)\left(\frac{\partial V^{i}}{\partial x^{j}}(X_{n}(t),t)-\frac{\partial V^{i}}{\partial x^{j}}(x,t)\right)\varphi(x-X_{n}(t))\right\Vert _{L^{p}}\\
  \leq&2\left\Vert A(t)\right\Vert ^{2}\left\Vert D^{2}K_{\varphi}\right\Vert _{\infty}\left\Vert \varphi\right\Vert _{L^{p}}r_{\varphi}
\end{align*}
In particular, if $K=G_{\delta}$, then $\nabla\wedge K_{\varphi}=0$,
so that 
\begin{align*}
\left\Vert F(\cdot,t)\right\Vert _{W^{-1,p}} & \leq6\left\Vert A(t)\right\Vert ^{2}\left\Vert DK_{\varphi}\right\Vert \left\Vert \varphi\right\Vert _{L^{p}}r_{\varphi}\\
 & \leq6C_{1}\frac{r_{\varphi}}{\delta^{3}}\left\Vert A(t)\right\Vert ^{2}\left\Vert \varphi\right\Vert _{L^{p}}
\end{align*}
and 
\[
\left\Vert G(x,t)\right\Vert _{L^{p}}\leq6C_{2}\frac{r_{\varphi}}{\delta^{4}}\left\Vert A(t)\right\Vert ^{2}\left\Vert \varphi\right\Vert _{L^{p}}
\]
The other conclusions follow from the following estimate: if $K=G_{\delta}$,
then $\left\Vert DK_{\varphi}\right\Vert \leq C_{1}\frac{1}{\delta^{3}}$,
thus by estimate (\ref{Aest3.7}) we obtain that 
\begin{align*}
\left\Vert A(t)\right\Vert  & \leq\left\Vert A(0)\right\Vert ^{2}e^{2\left\Vert DK_{\varphi}\right\Vert _{\infty}\left\Vert A(0)\right\Vert t}\\
 & \leq\left\Vert A(0)\right\Vert ^{2}e^{2C_{1}\frac{1}{\delta^{3}}\left\Vert A(0)\right\Vert t}\\
 & \leq\left\Vert A(0)\right\Vert ^{2}e^{C_{1}\frac{1}{\delta^{3}}\left\Vert A(0)\right\Vert }
\end{align*}
for all $t\in\left[0,\frac{1}{2}\tau_{A(0),\varphi}\wedge1\right]$. 
\end{proof}
We are now in a position to state a weak convergence theorem. Choose
$\delta>0$ so the solutions $\omega^{\delta}$ and $u^{\delta}$
solving (\ref{a-vort1}) and (\ref{a-vort2}) are good approximations
to the vorticity equations (\ref{vort-eq1}). Let $h>0$. Let $x_{n}$
where $n=(n^{1},n^{2},n^{3})\in\mathbb{Z}^{3}$ the center of the
lattice (open) box $B_{n}$ with size $h$ whose lower-left corner
is $hn$. Then $\omega_{0}$ may be approximated by $\sum_{n}\omega_{0}(x_{n})1_{B_{n}}$
in $L^{p}$ space for $p\geq1$. Let $W^{h}(x,0)=\sum_{n}A_{n,h}(0)\phi_{h}(x-x_{n})$
with $A_{n,h}(0)=\omega_{0}(x_{n})h^{3}$, where $n$ runs over all
$n$ such that $B_{n}\cap S\neq\textrm{Ã}$, and $\phi_{h}(x)=\frac{8}{h^{3}}\phi\left(\frac{2}{h}x\right)$,
where $\phi\geq0$ is chosen so that $\phi$ is an approximation of
$1_{B}$ in some Sobolev space, $\phi$ is smooth with a compact
support in $B$, where $B=(-1/2,1/2)^{3}$. Then 
\[
\left\Vert A_{,h}(0)\right\Vert =\sum_{n}|\omega_{0}(x_{n})|h^{3}\leq h|S|+\left\Vert \omega_{0}\right\Vert _{L^{1}}\leq|S|+\left\Vert \omega_{0}\right\Vert _{L^{1}}
\]
which is independent of $h\in(0,1)$ and the lattice size. With this
$W^{h}(x,0)$ as the initial sampling distribution, according to Corollary
\ref{cor5}, $V^{\delta,h},W^{\delta,h}$ defined by (\ref{sol-V-1})
and (\ref{sol-W-1}) with $K=G_{\delta}$ and $W(x,0)=W^{h}(x,0)$,
are approximations in mean of the vorticity equations (\ref{vort-eq1})
with initial vorticity $\omega_{0}$. Notice that 
\begin{align*}
K_{\phi_{h}}^{j}(x) & =\frac{8}{h^{3}}\int_{\mathbb{R}^{3}}K^{j}(y)\phi_{h}\left(\frac{2}{h}(x-y)\right)dy\\
 & =\frac{8}{h^{3}}\int_{\mathbb{R}^{3}}K^{j}(x-y)\phi_{h}\left(\frac{2}{h}y\right)dy\\
 & =\int_{|y|_{\infty}<\frac{1}{2}}K^{j}(x-\frac{h}{2}y)\phi\left(y\right)dy
\end{align*}
which yields that 
\[
\left\Vert D^{k}K_{\phi_{h}}\right\Vert _{\infty}\leq\left\Vert D^{k}K\right\Vert _{\infty}.
\]
Suppose $K=G_{\delta}$, then 
\[
\left\Vert D^{k}K_{\phi_{h}}\right\Vert _{\infty}\leq C_{k}\frac{1}{\delta^{2+k}}.
\]

\begin{thm}
Let $\delta>0$ and $h>0$. Let $\psi$, which is used to define $K=G_{\delta}$
for every $\delta>0$, and $\phi$, which is used to define the initial
$W^{h}(x,0)$ for every $h>0$, be two non-negative smooth functions
with supports lying inside the box $(-\frac{1}{2},\frac{1}{2})^{3}$
with $\int_{\mathbb{R}^{3}}\psi(x)dx=\int_{\mathbb{R}^{3}}\phi(x)dx=1$
as above. Let 
\begin{align}
&U^{\delta,h}(x,t)\nonumber\\=&W^{\delta,h}(x,t)-W^{h}(x,0)\nonumber\\-&\int_{0}^{t}\left(W^{\delta,h}\cdot\nabla V^{\delta,h}-V^{\delta,h}\cdot\nabla W^{\delta,h}+\nu\Delta W^{\delta,h}\right)ds.\label{U-eq1-1}
\end{align}
Then there is $T$ depending only on $|S|+\left\Vert \omega_{0}\right\Vert _{L^{1}}$
such that 
\[
\left\Vert \mathbb{E}\left[U^{\delta,h}(\cdot,t)\right]\right\Vert _{W^{-1,1}}\rightarrow0
\]
for all $t\in[0,T]$, as $h\downarrow0$ for any fixed $\delta>0$.
\end{thm}

\begin{proof}
This follows from the above estimates and Lemma \ref{lem7} item 2). 
\end{proof}
While from the proof we can see that the error term tends to zero
as $\delta\downarrow0$ and $h\downarrow0$ such that $h\ll\delta^{4}e^{-C/\delta^{3}}$
(with some positive constant $C$). This means that we need to choose
the lattice size $h$ much smaller than the regularization $\delta$
in order to ensure the convergence result. We should point out that this estimate is
very crude due to the lack of a priori estimates for solutions to the Navier-Stokes equations. The simulations in section 7 demonstrate that the convergence rate of our vortex method is much fast even for the case where $\delta=0$, but its proof is beyond the reach of the current mathematical analysis for these partial differential equations.

\section{Modified random vortex dynamics}

In the previous section, we have shown our random vortex system (\ref{eq:X-n eq2},
\ref{eq:A-n eq2}) converges to the vorticity equations in mean, while
if the viscosity $\nu>0$ is not small, then the random perturbation
term, i.e. the noise part involving Brownian motion, appearing in
(\ref{approx-vort1}) may be not small although its mean always stays
zero. To make the noise as small as possible, we may split each particle into $N$ copies of the same particle and apply independent copies of Brownian motion for these particles. Hence we propose the following SDEs
\begin{align}
dX_{n,a}^{i}(t)&=-\sum_{m,b}\varepsilon^{ikj}A_{m,a}^{k}(t)K_{\varphi}^{j}(X_{n,a}(t)-X_{m,b}(t))dt\nonumber \\ &+\sqrt{2\nu}dB_{n,a}^{i}(t),\nonumber\\ X_{n,a}(0)&=x_{n},\label{eq:X-n eq2-1}
\end{align}
and 
\begin{align}
&dA_{n,a}^{i}(t)=A_{n,a}^{l}(t)\sum_{m,b}\varepsilon^{ijk}\frac{\partial K_{\varphi}^{j}}{\partial x^{l}}(X_{n,a}(t)-X_{m,b}(t))A_{m,b}^{k}(t)dt,\nonumber\\ &A_{n,a}(0)=\frac{1}{N}A_{n}(0),\label{eq:A-n eq2-1}
\end{align}
where $a,b$ runs from $1$ up to $N$, where $N$ is a fixed natural
number, and $B_{n,a}$ are independent copies of 3D Brownian motion
on some probability space $(\varOmega,\mathcal{F},\mathbb{P})$. Thus
we have split the distribution of the initial vortices as 
\begin{align*}
W(x,0) & =\sum_{n}A_{n}(0)\varphi(x-x_{n})\\
 & =\sum_{n,a}A_{n,a}(0)\varphi(x-x_{n})
\end{align*}
and we have used independent Brownian motions for different locations
to reduce the eventual noise in the approximation vorticity equation
(\ref{approx-vort1}). $W(x,t)$ and $V(x,t)$ are still defined in
terms of (\ref{sol-V-1}) and (\ref{sol-W-1}), where the only modifications
one has to make are the sums over $n,a$ rather than $n$. The approximation
voticity equations (\ref{approx-vort1}) has the same form but with
different noise term: 
\[
dW^{i}=\left(W^{j}\frac{\partial V^{i}}{\partial x^{j}}-V^{j}\frac{\partial W^{i}}{\partial x^{j}}+\nu\Delta W^{i}\right)dt+H^{i}+\left(F^{i}+G^{i}\right)dt
\]
where $i=1,2,3$ and 
\[
H(x,t)=-\sqrt{2\nu}\sum_{n,a}A_{n,a}(t)\cdot\nabla\varphi(x-X_{n,a}(t))dB_{n,a}(t)
\]
which can be verified by using It\^{o}'s formula too. Hence 
\begin{align*}
&\mathbb{E}\left|\int_{0}^{t}dH(x,s)\right|^{2} \\ =&2\nu\frac{1}{N}\sum_{n}\int_{0}^{t}|A_{n}(s)|^{2}|\nabla\varphi(x-X_{n,a}(s))|^{2}ds\\
   \leq&\frac{4\nu}{N}\left\Vert \nabla\varphi\right\Vert _{\infty}^{2}\left\Vert A(0)\right\Vert \max_{n}|A_{n}(0)|e^{2\left\Vert DK_{\varphi}\right\Vert _{\infty}\left\Vert A(0)\right\Vert t}
\end{align*}
for $t\leq\frac{1}{2}\tau_{A(0),\varphi}$, which goes to zero as
$N\rightarrow\infty$. The error term estimates for $F$ and $G$
remain the same which are independent of $N$. Therefore
\[
\left\Vert\mathbb{E}\left[ |U^{\delta,h}(\cdot,t)|^2\right]\right\Vert _{W^{-1,1}}\rightarrow0
\]
for all $t\in[0,T]$, as $h\downarrow0$ for any fixed $\delta>0$.

\section{2D case}

We retain the same notations as used in the previous sections, but
we work with the two dimensional vorticity equation instead. 2D case
is special as we have indicated -- there is no non-linear stretching
term in the vorticity equation. The vorticity $\omega(x,t)$
of an incompressible fluid flow with velocity $u=(u^{1},u^{2})$ can
be identified with the scalar function $\frac{\partial}{\partial x^{1}}u^{2}-\frac{\partial}{\partial x^{2}}u^{1}$,
and the vorticity equation 
\begin{equation}
\frac{\partial}{\partial t}\omega+u^{j}\frac{\partial}{\partial x^{j}}\omega=\nu\Delta\omega\label{2D-vort1}
\end{equation}
appears as a ``linear'' parabolic equation. Hence the initial
vortices are transported along the motion of Brownian fluid particles.
Since $\textrm{div}u=0$, so that
\begin{equation}
\Delta u^{1}=-\frac{\partial}{\partial x^{2}}\omega,\quad\Delta u^{2}=\frac{\partial}{\partial x^{1}}\omega\label{2D-vort2}
\end{equation}
and
\[
u(x,t)=\int_{\mathbb{R}^{2}}G(x-y)\omega(y,t)dy
\]
where in 2D case, the Biot-Savart kernel $G(x)=\frac{1}{2\pi}(-\frac{x_{2}}{|x|^{2}},\frac{x_{1}}{|x|^{2}})$
which is singular at the original $0$. We therefore need to treat
it as for the 3D case to introduce $G_{\delta}(x)=G\ast\psi_{\delta}(x)$,
where $\psi_{\delta}(x)=\delta^{-2}\psi(\delta^{-1}x)$, and $\psi(x)$
is a smooth function with total integral $1$ and a compact support
lying inside $(-1/2,1/2)^{2}$. Then we have $\lVert DG_{\delta}(x)\rVert\leq C\frac{1}{\delta^{2}}$,
where $C$ depends only on the regularization $\psi$. We then implement
the vortex method to the approximation system \eqref{2D-vort1} together
with the integral relation
\begin{equation}
u(x,t)=\int_{\mathbb{R}^{2}}G_{\delta}(x-y)\omega(y,t)dy\label{2D-a-vort2}
\end{equation}
for each fixed $\delta>0$. Since $\omega(x,t)$ is scalar, so that
there is a simplification when we sample the distribution of initial
vortices. As for the 3D case, $(x_{n})$ is a finite collection of
points in $\mathbb{R}^{2}$ and the initial vortices $A_{n}(0)$ are
scalars. The dynamic equation for $A_{n}(t)$ is no longer needed
as $A_{n}(t)=A_{n}(0)=A_{n}$ are independent of $t$, and the dynamic
equation for $X_{n}(t)$ is reduced to the following SDE 
\begin{align}
&dX_{n}^{i}(t)=\sum_{m}A_{m}K_{\varphi}^{i}(X_{n}(t)-X_{m}(t)))dt+\sqrt{2\nu}dB^{i}(t),\nonumber\\&X_{n}(0)=x_{n},\label{eq: 2D-2}
\end{align}
whose coefficients are globally Lipschitz, so that it has a unique
strong solution $X_{n}(t)$ defined for all $t$. Let 
\begin{equation}
W(x,t)=\sum_{n}A_{n}\varphi(x-X_{n}(t))\label{eq:2Dvor}
\end{equation}
and 
\begin{equation}
V^{i}(x,t)=\sum_{n}A_{n}K_{\varphi}^{i}(x-X_{n}(t))\label{eq:2Dvel}
\end{equation}

By using It\^o's formula one may deduce that 
\[
\begin{aligned}&dW(x,t) \\ =&-V(x,t)\cdot\nabla W(x,t)dt+\nu\Delta W(x,t)dt\\-&\sum_{n}A_{n}\nabla\varphi(x-X_{n}(t))\cdot\sqrt{2\nu}dB(t)\\
  -&\sum_{n}A_{n}\nabla\varphi(x-X_{n}(t))\cdot(V(X_{n}(t),t)-V(x,t))dt.
\end{aligned}
\]
The non-martingale error term 
\[
F(x,t)=\sum_{n}A_{n}\nabla\varphi(x-X_{n}(t))\cdot(V(X_{n}(t),t)-V(x,t)).
\]
Following the same procedure as in the proof of Lemma \ref{lem7},
we may deduce the following global estimates for 2D case. 
\begin{lem}
\label{lem: 2d} 1) For any $p\geq1$ 
\[
\lVert F\rVert_{W^{-1,p}}\leq\lVert A\rVert^{2}\bigg(\lVert DK_{\varphi}\rVert\lVert\varphi\rVert_{L^{p}}r_{\varphi}+\lVert\varphi\text{\text{div}}K_{\varphi}\rVert_{L^{p}}\bigg).
\]
2) If the singular kernel $K=G_{\delta}$, where $\delta>0$, then
\[
\lVert F(\cdot,t)\rVert_{W^{-1,p}}\leq C_{5}\frac{1}{\delta^{2}}\lVert A\rVert^{2}\lVert\varphi\rVert_{L^{p}}r_{\varphi}
\]
for all $t>0$ where $\lVert A\rVert=\sum_{n}|A_{n}|$. 
\end{lem}

Now we state the weak convergence theorem. Let $h>0$, and let $x_{n}$
where $n=(n_{1},n_{2})\in\mathbb{Z}^{2}$ the center of the lattice
box $B_{n}$ with size $h$ whose lower left corner is $hn$. Then
the initial vorticity $\omega_{0}$ may be approximated by $\sum_{n}\omega_{0}(x_{n})1_{B_{n}}$
for example in an $L^{p}$-space. Let $W^{h}(x,0)=\sum_{n}A_{n,h}\varphi_{h}(x-x_{n})$
with $A_{n,h}=\omega_{0}(x_{n})h^{2}$, where $n$ runs over all $n$
such that $B_{n}\cap S\neq\tilde{A}$, and $\varphi_{h}(x)=\frac{1}{h^{2}}\phi(\frac{x}{h})$,
where $\phi$ is an approximation of $1_{B}$ where $B$ is the unit
square $(-1/2,1/2)^{2}$ in $L^{p}$ space for $p>1$. Then 
\[
\lVert A_{,h}\rVert=\sum_{n}|\omega_{0}(x_{n})|h^{2}\leq h|S|+\lVert\omega_{0}\rVert_{L^{1}}\leq|S|+\lVert\omega_{0}\rVert_{L^{1}},
\]
which is independent of $h\in(0,1)$ and the lattice size. With this
$W^{h}(x,0)$ as the initial sampling distribution, let $V^{\delta,h},W^{\delta,h}$
be defined as in equations (\ref{eq:2Dvel}) and (\ref{eq:2Dvor})
with $K=G_{\delta}$. Therefore, by Lemma \ref{lem: 2d}, item 2), we have
the following theorem: 
\begin{thm}
\label{thm2D}Let 
\begin{align}
U^{\delta,h}(x,t)&=W^{\delta,h}(x,t)-W^{h}(x,0)\nonumber \\&-\int_{0}^{t}\left(-V^{\delta,h}\cdot\nabla W^{\delta,h}+\nu\Delta W^{\delta,h}\right)ds.\label{U-eq1-1}
\end{align}
Then
\[
\left\Vert \mathbb{E}\left[U^{\delta,h}(\cdot,t)\right]\right\Vert _{W^{-1,1}}\rightarrow0
\]
for all $t>0$ as long as $h\ll\delta^{2}$, and for any $T>0$, this
convergence is uniform for all $(x,t)\in\mathbb{R}^{2}\times(0,T]$. 
\end{thm}

\section{Simulation results}

First we show some figures of the velocity field by our approximation scheme for 3D flows, which illustrates the velocity field projected at the plane $z=0$.
The flow starts with two points at $(-1,0,0)$ and $(1,1,0)$ with vorticity $(0,0,10)$ and $(10,10,10)$ respectively. i.e $x_1=(-1,0,0), x_2 = (1,1,0)$, $,A_1(0) = (0,0,10),A_2(0)=(10,10,10)$. Throughout the simulation in 3D flows, We choose the Biot-Savart kernel $G$, and we let $\varphi$  be $c(|x|^2-1)^2$, where $c$ is a constant such that the total integral of $\varphi = 1$ (We find a polynomial helps improve the computational speed of numerical integration in the packages we used in program). We let $h=0.5$. Figure.~\ref{fig:1} shows the case when $\nu=0$.

\begin{figure}
\centering
    \begin{subfigure}[t]{0.2\textwidth}
        \centering
        \includegraphics[width=1\linewidth]{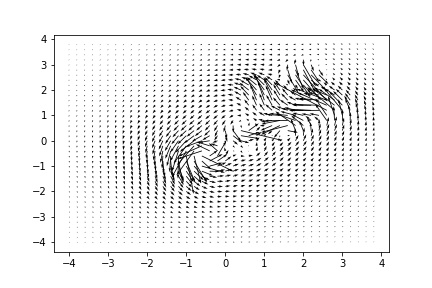} 
        \caption{$t=0$} 
    \end{subfigure} 
    \begin{subfigure}[t]{0.2\textwidth}
        \centering
        \includegraphics[width=1\linewidth]{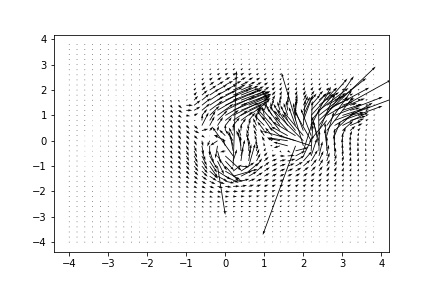} 
        \caption{$t=1$}
    \end{subfigure}
    \begin{subfigure}[t]{0.2\textwidth}
        \centering
        \includegraphics[width=1\linewidth]{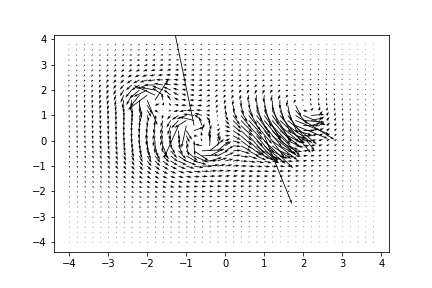} 
        \caption{$t=3$} 
    \end{subfigure}
    \begin{subfigure}[t]{0.2\textwidth}
           \centering
        \includegraphics[width=1\linewidth]{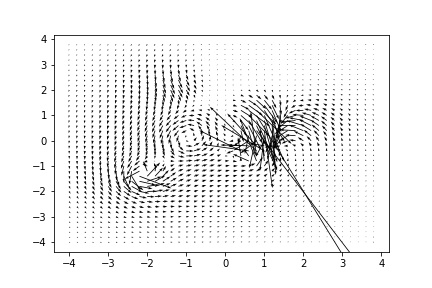} 
        \caption{$t=5$} 
    \end{subfigure}
    \caption{An inviscid Flow \label{fig:1}}

\end{figure}

\begin{figure}
    \centering
    \begin{subfigure}[h]{0.2\textwidth}
        \centering
        \includegraphics[width=1\linewidth]{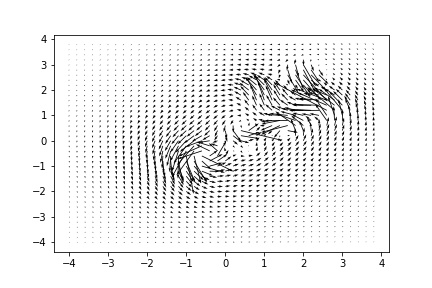} 
        \caption{} 
    \end{subfigure}
    \begin{subfigure}[h]{0.2\textwidth}
        \centering
        \includegraphics[width=1\linewidth]{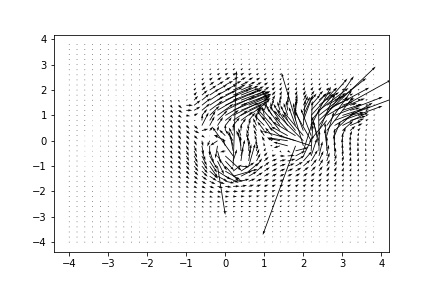} 
        \caption{}
    \end{subfigure}
    
    \begin{subfigure}[h]{0.2\textwidth}
        \centering
        \includegraphics[width=1\linewidth]{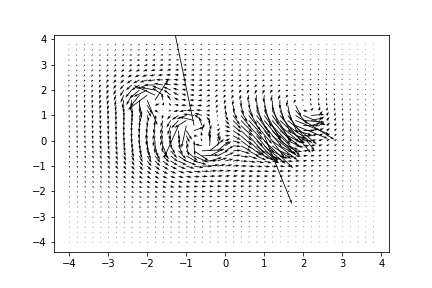} 
        \caption{} 
    \end{subfigure}
    \begin{subfigure}[h]{0.2\textwidth}
           \centering
        \includegraphics[width=1\linewidth]{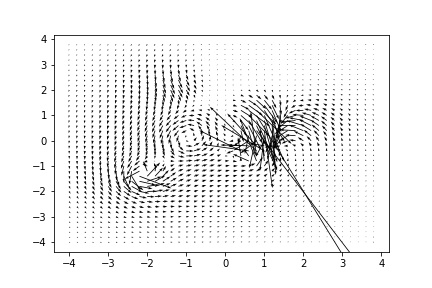} 
        \caption{} 
    \end{subfigure}
    \caption{Split each particle into $5$ parts\label{fig:2}}

\end{figure}

\begin{figure}
\centering
    \begin{subfigure}[h]{0.2\textwidth}
        \centering
        \includegraphics[width=1\linewidth]{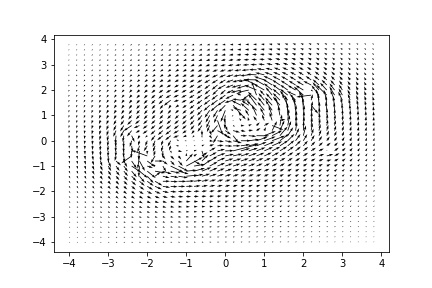} 
        \caption{} 
    \end{subfigure}
    \begin{subfigure}[h]{0.2\textwidth}
        \centering
        \includegraphics[width=1\linewidth]{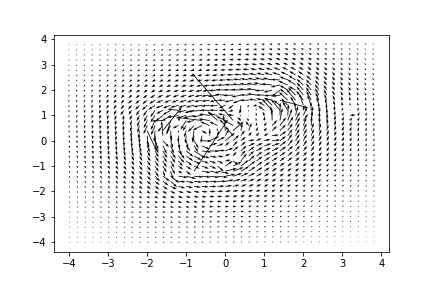} 
        \caption{}
    \end{subfigure}
    \begin{subfigure}[h]{0.2\textwidth}
        \centering
        \includegraphics[width=1\linewidth]{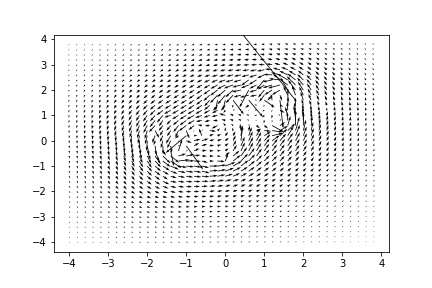} 
        \caption{} 
    \end{subfigure}
    \begin{subfigure}[h]{0.2\textwidth}
           \centering
        \includegraphics[width=1\linewidth]{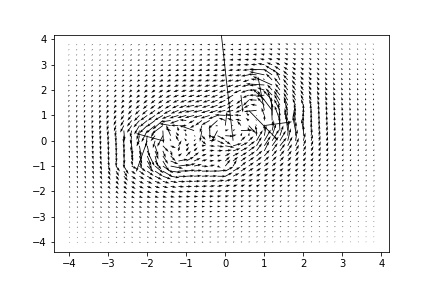} 
        \caption{} 
    \end{subfigure}
    \caption{Split each particle into $50$ parts\label{fig:3}}

\end{figure}

When $\nu\neq 0$, we are going to split the particles as in the previous section, and we will see that more splitting helps minimise the stochastic error. We keep the choice of the kernel, $\varphi$, $h$, the initial condition as in the inviscid flow.  The only difference here is that we consider the viscosity $\nu =1$ in this case, which means we have to solve a system of SDEs with the appropriate noise. We are looking at the projection of velocity field when $t = 0.5$. (Figure.~\ref{fig:2} and Figure.~\ref{fig:3}). When we split each particle into $50$ parts, the result from different randomisation is more 
consistent.

For 2D flows simulation, we use the 2D Biot-Savart kernel $H$, and $\varphi =c(|x|^2-1)^2$, $h=0.5$.  we start the flow with randomized $(x_n,A_n)$ $n =1,2,\dots, 10$. $x_n$ are sampled in $[-1,1]\times[-1,1]$ by uniform distribution independently, and $A_n$ are sampled in $[-10,-10]$ by uniform distribution. The particles inside the flow are denoted by different colors, and then we look at their locations at different time. (Figure.~\ref{fig:4})
\begin{figure}
    \centering
    \begin{subfigure}[h]{0.2\textwidth}
        \centering
        \includegraphics[width=1\linewidth]{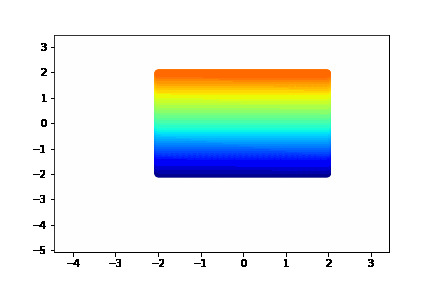} 
        \caption{$t=0$} 
    \end{subfigure}
    \begin{subfigure}[h]{0.2\textwidth}
        \centering
        \includegraphics[width=1\linewidth]{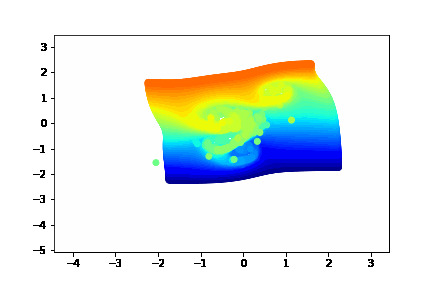}
        \caption{$t=0.3$}
    \end{subfigure}
    \begin{subfigure}[h]{0.2\textwidth}
        \centering
        \includegraphics[width=1\linewidth]{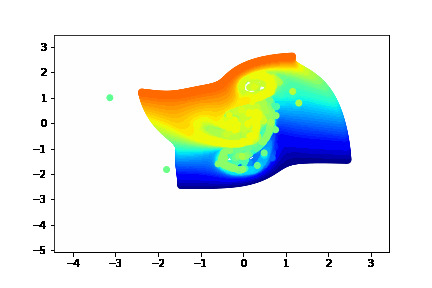}
        \caption{$t=0.6$} 
    \end{subfigure}
    \begin{subfigure}[h]{0.2\textwidth}
           \centering
        \includegraphics[width=1\linewidth]{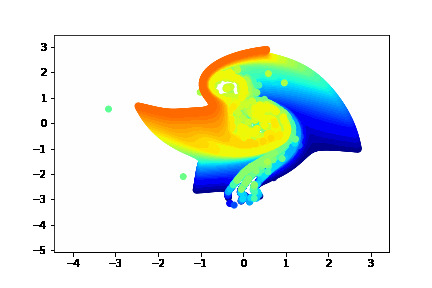}
        \caption{$t=1$} 
    \end{subfigure}
    \caption{The dynamics of an inviscid flow\label{fig:4}}

\end{figure}

Figure.~\ref{fig:5} is the vorticity picture for a non-inviscid flow with viscosity $\nu=1$. We use the same kernel, $\varphi$, $h$ as in the previous simulation. We let $(x_1,x_2,x_3,x_4)$ be sampled in $\mathbb{R}^2$ by $N(0,I)$, with vorticity $(10,10,-10,-10)$ respectively. We split each particle into $50$ parts to minimise the stochastic error. The vorticity value at each point is represented by its color, and the relations between the color and vorticity are shown in the bar on the right of each figure. The blue color represents positive vorticity, and the red color represents negative vorticity.
\begin{figure}
\centering
    \begin{subfigure}[h]{0.2\textwidth}
        \centering
        \includegraphics[width=1\linewidth]{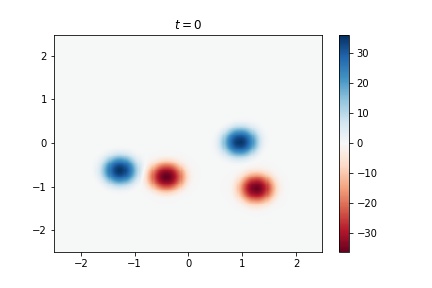}  
    \end{subfigure}
    \begin{subfigure}[h]{0.2\textwidth}
        \centering
        \includegraphics[width=1\linewidth]{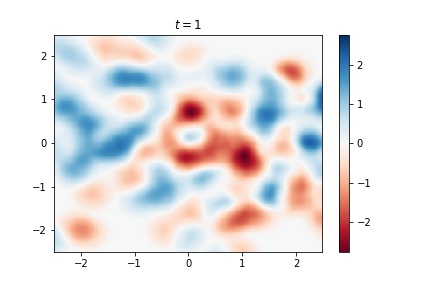} 
    \end{subfigure}
    \begin{subfigure}[h]{0.2\textwidth}
        \centering
        \includegraphics[width=1\linewidth]{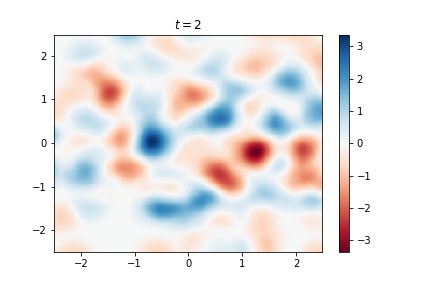} 
    \end{subfigure}
    \begin{subfigure}[h]{0.2\textwidth}
           \centering
        \includegraphics[width=1\linewidth]{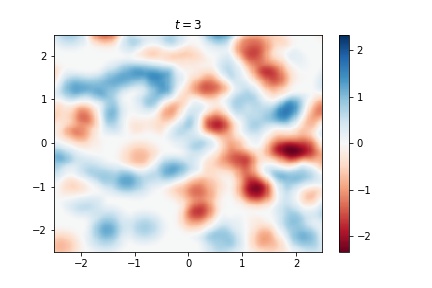} 
    \end{subfigure}
    \caption{Vorticity of a non-inviscid flow\label{fig:5}}

\end{figure}

\section*{Data Availability Statement}

The data that support the findings of this study are available from the corresponding author upon reasonable request.

\begin{acknowledgments}
This publication is based on work partially supported by the EPSRC Centre for Doctoral Training in Mathematics of Random Systems: Analysis, Modelling and Simulation (EP/S023925/1)
\end{acknowledgments}

\nocite{*}
\bibliography{bibliography}

\end{document}